\documentclass[journal,transmag]{IEEEtran}

\usepackage{enumerate}
\usepackage{amsmath}
\usepackage{amssymb,latexsym}
\usepackage{amsthm}
\usepackage{color}
\usepackage{graphicx}
\usepackage{hyperref}
\usepackage{amscd}
\usepackage{graphicx}
\usepackage{array}
\usepackage{amsfonts,mathrsfs}
\usepackage[flushleft]{threeparttable}
\usepackage{enumitem}
\usepackage{verbatim}

\DeclareMathOperator{\N}{N}

\title{On the list decodability  of Rank Metric codes}

\author{\IEEEauthorblockN{Rocco Trombetti\IEEEauthorrefmark{1} and
Ferdinando Zullo\IEEEauthorrefmark{2}}
\IEEEauthorblockA{\IEEEauthorrefmark{1}Dipartimento di Matematica e Applicazioni ``R. Caccioppoli", Universit\`{a} degli \\Studi di Napoli ``Federico II", I-80126 Napoli, Italy}
\IEEEauthorblockA{\IEEEauthorrefmark{2}Dipartimento di Matematica e Fisica, Universit\`{a} degli \\Studi della Campania ``Luigi Vanvitelli", I-81100 Caserta, Italy }
\thanks{This work is supported by the Research Project of MIUR (Italian Office for University and Research) ``Strutture Geometriche, Combinatoria e loro Applicazioni" 2012. The second author was also supported by the project ''VALERE: Vanvitelli pEr la RicErca" of the University of Campania ''Luigi Vanvitelli.''
	
Corresponding author: R. Trombetti (email: rtrombet@unina.it) }}

\newcommand{\cG}{{\mathcal G}}

\newcommand{\cH}{{\mathcal H}}
\newcommand{\cD}{{\mathcal D}}

\newcommand{\cS}{{\mathcal S}}

\newcommand{\F}{{\mathbb F}}

\newcommand{\Z}{{\mathbb Z}}

\newcommand{\End}{\mathrm{End}}

\newcommand{\la}{\langle}
\newcommand{\ra}{\rangle}
\renewcommand{\mod}{\hbox{{\rm mod}\,}}

\newtheorem{theorem}{Theorem}
\newtheorem{lemma}[theorem]{Lemma}
\newtheorem{example}[theorem]{Example}
\newtheorem{corollary}[theorem]{Corollary}
\newtheorem{proposition}[theorem]{Proposition}

\theoremstyle{definition}

\newtheoremstyle{TheoremNum}
{\topsep}{\topsep}              
{\itshape}                      
{}                              
{\bfseries}                     
{.}                             
{ }                             
{\thmname{#1}\thmnote{ \bfseries #3}}
\newtheorem{remark}{Remark}

\DeclareMathOperator{\PG}{\mathrm{PG}}

\begin{document}

\maketitle

\begin{abstract} Let $k,n,m \in \Z^+$ be integers such that $k\leq n \leq m$, let $\mathrm{G}_{n,k}\in \F_{q^m}^n$ be a Delsarte-Gabidulin code (\cite{delsarte_bilinear_1978},\cite{gabidulin_MRD_1985}).

In \cite{wachter-zhe_2013}, Wachter-Zeh proved that codes belonging to this family cannot be efficiently list decoded for any radius $\tau$, providing $\tau$ is large enough. This achievement essentially relies on proving a lower bound for the list size of some specific words in $\F_{q^m}^n \setminus \mathrm{G}_{n,k}$;  \cite[Theorem 1]{ wachter-zhe_2013}.

In \cite{raviv_2016}, Raviv and Wachter-Zeh improved this bound in a special case, i.e. when $n\mid m$. As a consequence, they were able to detect infinite families of Delsarte-Gabidulin codes that cannot be efficiently list decoded at all.

In this article we determine similar lower bounds for Maximum Rank Distance codes belonging to a wider class of examples, containing Generalized Gabidulin codes, Generalized Twisted Gabidulin codes, and examples recently described by the first author and Yue Zhou in \cite{trombetti_new}. By exploiting arguments such like those used in \cite{wachter-zhe_2013} and \cite{raviv_2016}, when $n\mid m$, we also show infinite families of generalized Gabidulin codes that cannot be list decoded efficiently at any radius greater than or equal to $\left\lfloor \frac{d-1}2 \right\rfloor+1$, where $d$ is its minimum distance. Nonetheless, in all the examples belonging to above mentioned class, we detect infinite families that cannot be list decoded efficiently at any radius greater than or equal to $\left\lfloor \frac{d-1}2 \right\rfloor+2$, where $d$ is its minimum distance. In particular, this leads to show infinite families of Gabidulin codes, with underlying parameters not covered by \cite[Theorem 4]{raviv_2016}, having this decodability defect.

Finally, relying on the properties of a set of subspace trinomials recently presented in \cite{McGuireMueller}, we are able to prove our main result, that is any rank-metric code of $\F_{q^m}^n$ of order $q^{kn}$ with $n$ dividing $m$, such that $4n-3$ is a square in $\mathbb{Z}$ and containing $\mathrm{G}_{n,2}$, is not efficiently list decodable at some values of the radius $\tau$.
\end{abstract}

\section{Introduction}
Let $q$ be a prime power. The set $\F_q^{m\times n}$ of all $m\times n$ matrices over $\F_q$, is an $\F_q$-vector space. The \emph{rank metric distance} on $\F_q^{m\times n}$ is defined by
\[d(A,B)=\mathrm{rk}(A-B) \,\, \text{for} \,\, A,B\in \F_q^{m\times n},\]
where $\mathrm{rk}(C)$ stands for the rank of $C$.

A subset $\mathrm{C}\subseteq \F_q^{m\times n}$ is called a \emph{rank-metric code}. The \emph{minimum distance} of $\mathrm{C}$ is
\[d(\mathrm{C})=\min_{A,B\in \mathrm{C}, A\neq B} \{d(A,B)\}.\]
When $\mathrm{C}$ is a subspace of $\F_q^{m\times n}$, we say that $\mathrm{C}$ is an $\F_q$-linear code and its dimension $\dim_{\F_q}(\mathrm{C})$ is defined to be the dimension of $\mathrm{C}$ as a subspace over $\F_q$. It is well-known that
\begin{equation}\label{singletonbound}
|\mathrm{C}| \le q^{\max\{m,n\}(\min\{m,n\}-d(\mathrm{C})+1)}.
\end{equation}
When the equality holds, we call $\mathrm{C}$ a \emph{maximum rank distance} (MRD for short) code.

For MRD codes with minimum distance less than $\min\{m,n\}$, there are a few known constructions. The first and most famous family is due to Delsarte \cite{delsarte_bilinear_1978} and Gabidulin \cite{gabidulin_MRD_1985} who found it independently. This family is later generalized by Kshevetskiy and Gabidulin in \cite{kshevetskiy_new_2005}, and we often call them \emph{generalized Gabidulin codes}. More recently in \cite{sheekey_new_2016}, the author exhibited two infinite families of linear MRD codes which are not equivalent to generalized Gabidulin codes. We call them {\it twisted Gabidulin codes} and {\it generalized twisted Gabidulin} codes. In \cite{lunardon_generalized_2015} it was shown that the latter family contains both generalized Gabidulin codes and twisted Gabidulin codes as proper subsets, and in \cite{OzbudakOtal} this family was further generalized. Finally in \cite{trombetti_new} the authors presented a new family of MRD codes in the case when $m$ is even.
Further families of MRD codes are known for some particular values of the parameters \cite{BZZ,CMPZ,CsMPZh,CsMZ2018,MMZ}.

\medskip
\noindent Consider the set of $q$-polynomials with coefficients in $\F_{q^m}$; i.e., the set of polynomials defined as follows: \[\mathcal{L}_{m,q}[x]=\left\{\sum_{i=0}^{l} c_i x^{q^i}: c_i\in \F_{q^m} \,\, \ell \in \mathbb{N}_0 \right\}.\]

Any polynomial $f$ in $\mathcal{L}_{m,q}[x]$ gives rise to an $\F_q$-linear map $x\in \F_{q^m} \mapsto f(x) \in \F_{q^m}$. If $c_l \neq 0$ we will refer to $l$ as to the $q$-degree of $f$.

It is well known that $$(\mathcal{L}_{m,q}[x] / (x^{q^m}-x),+,\circ,\cdot),$$ where $+$ is the addition of maps, $\circ$ is the composition of maps and $\cdot$ is the scalar multiplication by elements of $\F_q$, is isomorphic to the algebra of $m\times m$ matrices over $\F_q,$ and hence to $\End_{\F_q}(\F_{q^m})$; i.e., the set of endomorphisms on $\F_{q^m}$ seen as an $\F_q$-algebra. In the following we will denote this algebra by the symbol $\tilde{\mathcal{L}}_{m,q}[x]$ and we will always silently identify the elements of $\tilde{\mathcal{L}}_{m,q}[x]$ with the endomorphisms of $\F_{q^m}$ they represent. Consequently, we will speak also of \emph{kernel} and \emph{rank} of a polynomial meaning by this the kernel and rank of the corresponding endomorphism.
Clearly, the kernel of $f\in \tilde{\mathcal{L}}_{m,q}[x]$ coincides with the set of the roots of $f$ over $\F_{q^m}$ and as usual $\dim_{\F_q} \mathrm{Im}(f)+\dim_{\F_q} \ker(f)=m$.

Let $\sigma$ be a generator of the Galois group $\mathrm{Gal}(\F_{q^m}\colon\F_q)$. A $\sigma$-linearized polynomial over $\F_{q^m}$, is a polynomial of type $f(x)=c_0x+c_1x^{\sigma}+\cdots+c_{\ell}x^{\sigma^{\ell}}$ with $c_j \in \F_{q^m}$. If $c_{\ell} \neq 0$, we say that $\ell$ is the $\sigma$-degree of the polynomial and we will denote it by $\deg_{\sigma} f$. It is clear that a $\sigma$-linearized polynomial can always be seen as an element of $\mathcal{L}_{m,q}[x]$.

From now on suppose $n \leq m$. All above mentioned examples produce MRD codes that can be presented in terms of so called {\it puncturing} operation applied to a subset of $\mathcal{L}_{m,q}[x]$.

Precisely, let $f_1$ and $f_2$ be two additive functions of $\F_{q^m}$ and let $k \leq m-1$. Following \cite[Proposition 1]{sheekey_newest_preprint}, the subset of $\mathcal{L}_{m,q}[x]$

\begin{small}
\begin{equation}\label{eq:codes}
\cH_{m,k,{\sigma}}(f_1,f_2)=\left\{f_1(a)x + \sum_{i=1}^{k-1}a_ix^{\sigma^i}+f_2(a)x^{\sigma^k} \colon a,a_i \in \F_{q^m}\right\}
\end{equation}
\end{small}

with $N(f_1(a)) \neq (-1)^{mk}N(f_2(a))$ for all $a\in \F_{q^m}^*$, where $N(f(a))= \N_{q^m/q}(f(a))=f(a)^{1+q+\ldots+q^{m-1}}$, defines an MRD code of $\F_q^{m \times m}$ with minimum distance $d=m-k+1$. For instance, if $f_1(a)=a$ and $f_2=0$, then $\cH_{m,k,\sigma}(f_1,f_2)$ is a generalized Gabidulin code (commonly indicated with the symbol $\cG_{m,k,\sigma}$).

More in detail, Table \ref{table:ex} reports the examples of the so far known MRD codes of $\tilde{\mathcal{L}}_{m,q}[x]$ that can be represented as in (\ref{eq:codes}).

\begin{table*}[ht]\label{table:ex}
\caption{Known examples of MRD codes in $\tilde{\mathcal{L}}_{m,q}[x]$}
\centering
\begin{tabular}{c c c c c c c}
\hline\hline
\text{Symbol} & $\sigma$ & $f_1(a)$ & $f_2(a)$ & \text{Conditions} & \text{References} \\
\hline
$\cG_{m,k}$ & $q$      &  $a$           & 0             & --- & \cite{delsarte_bilinear_1978,gabidulin_MRD_1985} \\
$\cG_{m,k,\sigma}$ & $q^s$      &  $a$           & 0             & --- & \cite{kshevetskiy_new_2005} \\
$\cH_{m,k,\sigma}(\eta,h)$ &  $q^s$     &  $a$           & $\eta a^{q^h}$  & $N_{q^m/q}(\eta) \neq (-1)^{mk}$ & \cite{sheekey_new_2016,lunardon_generalized_2015} \\
$\overline{\cH}_{m,k,\sigma}(\eta,h)$ &  $q^s$     &  $a$           & $\eta a^{p^h}$  & $N_{q^m/p}(\eta) \neq (-1)^{mk}$ & \cite{OzbudakOtal} \\
$\cD_{m,k,\sigma}(\eta)$    &  $q^s$     &   $a$     &    $\eta b$ & $m$ \text{even}, $N_{q^m/q}(\eta) \notin \square$, $a,b \in \F_{q^{m/2}}$ & \cite{trombetti_new} \\
\hline
\end{tabular}
\\ where $\N_{q^m/q}(\eta)=\eta^{1+q+\ldots+q^{m-1}}$, $\N_{q^m/p}(\eta)=\eta^{1+p+\ldots+p^{m\ell-1}}$ ($q=p^\ell$) and $\square$ is the set of square elements in $\F_{q^m}$.
\end{table*}

All of the examples above contain a subcode equivalent to $\cG_{m,k-1,\sigma}$ and if $\cH_{m,k,\sigma}(f_1,f_2)$ is not equivalent to any generalized Gabidulin code, then $k-1$ is the maximum dimension of a subcode in it which is equivalent to a generalized Gabidulin code, i.e. $\cH_{m,k,\sigma}(f_1,f_2)$ has {\it Gabidulin index} $k-1$, (see \cite{GZ}).

There are other equivalent ways of representing a rank-metric code of $\F_q^{m\times n}$. For our purpose, we will see such codes also as subsets of $\F_{q^m}^n$.

For a vector $\mathbf{v}=(v_1,\ldots,v_n)\in \F_{q^m}^n$, we define its rank weight as follows
\[ \mathrm{rk}(\mathbf{v})=\dim_{\F_q}  \la v_1,\ldots, v_n \ra_{\F_q}. \]
The rank distance between two vectors $\mathbf{u}, \mathbf{v} \in \F_{q^m}^n$ is defined as $d(\mathbf{u},\mathbf{v})=\mathrm{rk}(\mathbf{u}-\mathbf{v})$.
A rank-metric code of $\F_{q^m}^n$ is a subset of $\F_{q^m}^n$ equipped with the aforementioned metric.
The same bound \eqref{singletonbound} holds and hence we can define again an MRD code $\mathrm{C}$ as the code whose parameters attain the equality in \eqref{singletonbound}, i.e. $|\mathrm{C}|=q^{mk}$ and for each $\mathbf{u},\mathbf{v}\in \mathrm{C}$ with $\mathbf{u}\neq\mathbf{v}$ we have that $\mathrm{rk}(\mathbf{u}-\mathbf{v})\geq n-k+1$.

Of course we may always jump from the model of linearized polynomials to $\F_{q^m}^m$. More in general we have the following

\begin{lemma}\label{lm:puncturing}
Let $\mathcal{C}$ be an MRD code of $\tilde{\mathcal{L}}_{m,q}$ with $|\mathcal{C}|=q^{mk}$ and $d(\mathcal{C})=m-k+1$.
Let $n$ be a positive integer greater than or equal to $k$ and the $\mathcal{S}=\{\alpha_1,\ldots,\alpha_n\}$ be an $n$-set in $\F_{q^m}$ (i.e.  $n$ $\F_q$-linearly independent elements).
Then the rank-metric code
\[ \mathrm{C}=\{(g(\alpha_1),\ldots,g(\alpha_n)) \colon g \in \mathcal{C}\}\subseteq \F_{q^m}^n \]
is an MRD code of $\F_{q^m}^n$ with $|\mathcal{C}|=q^{mk}$ and $d(\mathrm{C})=n-k+1$.
\end{lemma}

As a consequence of this lemma, we have that
\begin{equation}\label{eq:general_code_from_poly}
	\mathrm{H}_{n,k,\sigma}(f_1,f_2)[\underline{\alpha}]=\left\{\left(g(\alpha_1), \dots, g(\alpha_n)\right): g\in \cH_{m,k,\sigma}(f_1,f_2) \right\},
\end{equation}
where $\underline{\alpha}=(\alpha_1,\ldots,\alpha_n)$ with $\dim_{\F_{q^m}}\langle \alpha_1,\ldots,\alpha_n\rangle_{\F_{q^m}}=n$,
gives rise to an  MRD code of $\F_{q^m}^n$ with minimum distance $d=n-k+1$.

In the following, when $\sigma$, $k$ and $\underline{\alpha}$ are clear from the context, we denote $\mathrm{H}_{n,k,\sigma}(f_1,f_2)[\underline{\alpha}]$ by $\mathrm{H}_{n}(f_1,f_2)$.

Codes $H_{n,k,1}(a,0)=\{\left(g(\alpha_1), \dots, g(\alpha_n)\right): g \in \mathcal{G}_{m,k}\}$, correspond to those presented in \cite{delsarte_bilinear_1978} and \cite{gabidulin_MRD_1985}, and are known as Delsarte-Gabidulin codes. We will denote them in the following by the symbol $\mathrm{G}_{n,k}$.
In particular, $\mathrm{G}_{n,k,\sigma}[\underline{\alpha}]$ denotes the rank-metric code obtained by evaluating each element of $\mathcal{G}_{m,k,\sigma}$ in $\underline{\alpha} \in \F_{q^m}^n$.

\medskip
\noindent For each element $\omega \in \F_{q^m}^n$ and $\tau \in \Z^+$, we define
$$B_{\tau}(\omega):= \{c \in \F_{q^m}^n \, | \,  \mathrm{rk}(\omega-c) \leq \tau\}.$$

\smallskip

Elias in \cite{Elias} and Wozencraft in \cite{Wozencraft} introduced, in the Hamming metric, the problem of list decoding a given code. Such a problem may be stated in the following very general fashion. Let $C$ be any code of lenght $n$ in the metric space $\F_{q^m}^n$, and let $\tau $ be a positive number (a radius). Given a received word, output the list of all codewords of the code within distance $\tau$ from it.   A \emph{list decoding algorithm} returns the list of all codewords with distance at most $\tau$ from any given word.

We say that $\mathrm{C}\subseteq \F_{q^m}^n$ is \emph{efficiently list decodable} at the radius $\tau$, if there exists a polynomial-time (in the length of the code, i.e. $n$) list decoding algorithm. Of course, if there exists a word $\omega \in \F_{q^m}^n \setminus \mathrm{C}$ for which $B_{\tau}(\omega)\cap \mathrm{C}$ has exponential size in $n$, such an algorithm cannot exist since writing down the list already has exponential complexity. When such an algorithm does not exist we say that $C$ is not efficiently list decodable at the radius $\tau$.
See \cite{guruswami} for further details on the list decodability issue.

It is well known that many of the codes in \eqref{eq:general_code_from_poly} can be efficiently decoded whenever up to $\left\lfloor \frac{d-1}2 \right\rfloor$ rank errors occur, where $d$ is its minimum distance. Several decoding algorithms exist for Gabidulin codes as shown by Gabidulin in \cite{gabidulin_MRD_1985}, by Richter and Plass in \cite{RP04} and by Loidreau in \cite{Loi06}. These methods were speeded up by Afanassiev, Bossert, Sidorenko and Wachter-Zeh in \cite{WAS11,WAS13,WSB10} and more recently by Puchinger and Wachter-Zeh in \cite{PWZ}, see also \cite{WT}. First Randrianarisoa and Rosenthal in \cite{RandrianarisoaRosenthal} and then, completing missing cases, Randrianarisoa in \cite{Randrianarisoa} proposed a decoding algorithm for generalized twisted Gabidulin codes, see also \cite{Li} for the additive case and for fields of characteristic two.

However, in general it is not clear whether these codes, as well as others in Class \eqref{eq:general_code_from_poly}, can be efficiently list decoded from a larger number of errors.

In \cite{wachter-zhe_2013}, by adapting to the rank metric setting techniques appeared in \cite{32,31}, the author proved that Delsarte-Gabidulin codes $\mathrm{G}_{n,k}\subset \F_{q^m}^n$ with minimum distance $d$ cannot be efficiently list decoded at any radius $\tau$ such that \begin{equation}\label{eq:boundontau} \tau \geq \frac{m+n}{2}-\sqrt{\frac{{(m+n)}^2}{2} -m(d-\epsilon)},\end{equation} where $0\leq \epsilon  < 1$.

More precisely, as a corollary of \cite[Theorem 1]{wachter-zhe_2013},  an exponential lower bound for the size of $\mathrm{G}_{n,k} \cap B_{\tau}(\omega)$ was proven, for suitable elements $\omega$ in $\mathbb{F}_{q^m}^n$, providing (\ref{eq:boundontau}) holds true.

In \cite{raviv_2016} the authors improved this result under specific restrictions for the involved parameters. As a consequence, infinite families of Delsarte-Gabidulin codes are explicitly exhibited, which are not efficiently list decodable for each $\tau$ exceeding the unique decoding radius by one. In other words codes in these latter families, cannot be list decoded efficiently at all.

In this article we prove a similar limit in list decoding behavior for all others examples in the Class described in \eqref{eq:general_code_from_poly}.

Precisely, elaborating on the techniques used in \cite{wachter-zhe_2013} and in \cite{raviv_2016}, we generalize to these latter examples, results contained in \cite{wachter-zhe_2013} and \cite{raviv_2016}. As a consequence of this, providing that $n$ divides $m$, infinite families of generalized Gabidulin codes that cannot be efficiently list decoded at all, are detected. Also, we exhibit infinite families of codes in $\mathrm{H}_{n}(f_1,f_2)$, that cannot be list decoded efficiently at any radius greater than or equal to $\tau=\left\lfloor \frac{d-1}2 \right\rfloor+2$.

Finally, relying on the properties of a set of subspace of $q$-trinomials recently presented in \cite{McGuireMueller}, we are able to prove that any rank-metric code of $\F_{q^m}^n$ of order $q^{kn}$ with $n$ dividing $m$ containing $\mathrm{G}_{n,2}$ (which implies that its Gabidulin index is at least two), and such that $4n-3$ is a square in $\mathbb{Z}$, is not efficiently list decodable for any radius greater than or equal to $\frac{2n-1-\sqrt{4n-3}}{2}$.

\section{Preliminary results}
\noindent Let $\cS$ be an $n$-subset of $\F_{q^m}$, and let $U_\cS$ be the $\F_q$-subspace of $\F_{q^m}$, seen as $m$-dimensional vector space over $\F_q$, generated by the elements of $\cS$.

In order to investigate maximum rank-metric codes in $\F_{q^m}^n$ where $n < m$ in terms of polynomials, we will need the following results concerning with  $q$-polynomials over $\F_{q^m}$.

\begin{lemma}\label{le:polynomials_matrices}
	Let $n$, $m$ be in $\Z^+$ satisfying that $n\leqslant m$, and let $q$ be a prime power. Let $\cS$ be a subset consisting of $n$ arbitrary $\F_q$-linearly independent elements $\alpha_1,\dots,\alpha_n\in \F_{q^m}$. Define
	$\theta_\cS:= \prod_{u\in U_\cS} (x-u)$.
	Then we have
	\[\mathcal{L}_{m,q}[x] / (\theta_\cS) \cong \left\{\left(f(\alpha_1), \dots,f(\alpha_n)\right): f\in \mathcal{L}_{m,q}[x] \right\} .\]
\end{lemma}
\begin{proof}
	The map given by
	\[\varphi : f\in \mathcal{L}_{m,q}[x] \mapsto  \left(f(\alpha_1), \dots, f(\alpha_n)\right)\in \F_{q^m}^n\]
	is clearly surjective and $\F_q$-linear. By noting that $\varphi(f)$ is the zero vector if and only if $f(x)=0$ for every $x\in U_\cS$, we see that $\ker(\varphi) = \{f\, \in \mathcal{L}_{m,q}[x]: f \equiv 0 \, \mod{\theta_\cS}\}$. This concludes the proof.
\end{proof}

For the subset $\cS$ made up of $n$ arbitrary $\F_q$-linearly independent elements in $\F_{q^m}$, we define
\begin{center}
\begin{tabular}{cccc}
	$\pi_\cS$ :& $\mathcal{L}_{m,q}[x]$ &$\rightarrow$ &$\mathcal{L}_{m,q}[x]/(\theta_\cS)$,\vspace{0.2cm}\\
		 & $f$            &$\mapsto$     & $f \, \mod{\theta_\cS}$.
\end{tabular}
\end{center}
In particular, as already observed, when $U_\cS=\F_{q^m}$, by Lemma \ref{le:polynomials_matrices} we have
\[ \End_{\F_q}(\F_{q^m}) \cong  \mathcal{L}_{m,q}[x]/(x^{q^m}-x).\]

\begin{lemma}\label{le:f=0}
Let $\cS$ be an $n$-subset of \,$\F_q$-linearly independent elements in $\F_{q^m}$. Let $\mathcal{C}$ be a subset of $\mathcal{L}_{m,q}[x]$. Assume that for any distinct $f$ and $g\in \mathcal{C}$, the number of solutions of $f=g$ in $U_\cS$ is strictly smaller than $q^n$. Then $\pi_\cS$ is injective on $\mathcal{C}$.
\end{lemma}
\begin{proof}
It follows directly from the assumption $|\{x\in \F_{q^m}: f(x)=0 \}|<q^n=|U_\cS|$ and the fact that $f\equiv 0 \,\, \mod \, { \theta_\cS}$ if and only if $f(u)=0$ for every $u\in U_\cS$.
\end{proof}

By Lemma \ref{le:f=0} and the fact that  $\cG_{n,k,\sigma}$ is an MRD code, the following result can readily be verified.
\begin{corollary}\label{coro:representation_all}
Let $\cS$ be an $n$-subset of $\F_q$-linearly independent elements in $\F_{q^m}$. Let $\sigma$ be a generator of $\mathrm{Gal}(\F_{q^m}\colon\F_q)$. Then the set
\begin{equation}\label{tran}
\mathrm{Tran}=\{a_0 x + a_1 x^{\sigma} + \dots +a_{n-1} x^{\sigma^{^{n-1}}}: a_i\in \F_{q^m} \}
\end{equation}
is a transversal, namely a system of distinct representatives, for the ideal $(\theta_\cS)$ in  $\mathcal{L}_{m,q}[x]$.
\end{corollary}

This implies that if $f(x)\in \mathcal{L}_{m,q}[x]$ is a $\sigma$-polynomial with $\sigma$-degree less or equal to $n-1$, then
\[ f(x) \equiv 0 \pmod{\theta_\cS} \Leftrightarrow f(x)= 0, \]
and if $f(x)$ and $g(x)$ are two $\sigma$-polynomials in $\mathrm{Tran}$ then
\[f(x) \equiv g(x) \pmod{\theta_\cS} \Leftrightarrow f(x)=g(x). \]


Clearly, if an $\F_q$-linear subset $\mathcal{C}\subseteq \mathcal{L}_{m,q}[x]$ is  of size $q^{nk}$ and each nonzero polynomial in $\mathcal{C}$ has at most $q^{k-1}$ roots over $U_\cS$, for instance $\mathcal{C}=\cH_{m,k,{\sigma}}(f_1,f_2)$, then the assumption on $\mathcal{C}$ in Lemma \ref{le:f=0} is satisfied.

By Lemmas \ref{le:polynomials_matrices} and \ref{le:f=0}, it follows that codes described in \eqref{eq:general_code_from_poly} can be equivalently written as
\begin{small}
\begin{equation}\label{eq:polynomials_GTGC}
		\pi_\cS(\cH_{m,k,{\sigma}}(f_1,f_2)) = \{f \, \mod {\theta_\cS}: f\in \cH_n(f_1,f_2)\}\subseteq \mathcal{L}_{m,q}[x]/(\theta_\cS).
\end{equation}
\end{small}
In particular, when $n=m$, it becomes
$$\{f \, \mod (x^{q^m}-x): f\in \cH_{k,m,{\sigma}}(f_1,f_2)\}\subseteq \tilde{\mathcal{L}}_{m,q}[x].$$

\section{Bounds on list decodability of MRD codes in $\mathrm{H}_{n}(f_1,f_2)$}\label{sec:bounds}

Let $\sigma$ be a generator of $\mathrm{Gal}(\F_{q^m}\colon\F_q)$. In what follows the concept of {\it $\sigma$-subspace polynomial} will be of some importance. A $\sigma$-\emph{subspace polynomial} with respect to $\F_{q^m}$, is a {\it monic} linearized polynomial, say $s(x)$, satisfying the property that, if $r=\deg_{\sigma} s$, there exists an $r$-dimensional subspace $U$ of $\F_{q^m},$ seen as a vector space over $\F_q$, such that
\begin{small}
$$s(x)=(-1)^{r+1}\frac{1}{\left| \begin{pmatrix}
u_1 & u_1^\sigma & \cdots & u_1^{\sigma^{r-1}} \\
\vdots \\
u_r & u_r^\sigma & \cdots & u_r^{\sigma^{r-1}}
\end{pmatrix} \right|}  \begin{pmatrix}
x & x^\sigma & \cdots & x^{\sigma^r} \\
u_1 & u_1^\sigma & \cdots & u_1^{\sigma^r} \\
\vdots \\
u_r & u_r^\sigma & \cdots & u_r^{\sigma^r}
\end{pmatrix},$$
\end{small}
where $u_1,\ldots,u_r$ is an $\F_q$-basis of $U$.

\noindent Clearly, $s(x)=a_0x+a_1x^{\sigma}+\cdots+a_{r-1}x^{{\sigma}^{r-1}}+x^{{\sigma}^r}$ for some $a_0,a_1,...,a_{r-1} \in \F_{q^m}.$ Also, we have that each subspace of $\F_{q^m}$ corresponds to a unique $\sigma$-subspace polynomial.

Denote by $\mathcal{P}_{r,\sigma} \subset  \mathcal{L}_{m,q}[x]$ the set of all $\sigma$-subspace polynomials in $ \mathcal{L}_{m,q}[x]$ associated with $r$-dimensional subspaces of $\F_q^m$. Clearly, we have \[|\mathcal{P}_{r,\sigma}|={m\brack r}_q.\]
Let $\cS$ be an $n$-subset of $\F_q$-linearly independent elements in $\F_{q^m}$ and suppose that $r < n$. As a consequence of Corollary \ref{coro:representation_all} and since $r<n$, it is also plain that $|\pi_\cS(\mathcal{P}_{r,\sigma})|={m\brack r}_q$.

Arguing as in \cite[Theorem 1]{raviv_2016}, we may now show the existence of a {\it large set} of $\sigma$-subspace polynomials in $\mathcal{L}_{m,q}[x]$ agreeing on their top-most $\sigma$-coefficients, whose kernels are contained in a fixed $n$-dimensional subspace. More precisely,


\begin{lemma}\label{lemma:subspaceset}
Let $g,r,n$ and $m \in \mathbb{Z}^+$ be positive integers such that $g \leq r < n \leq m$.
Let $\cS$ be an $n$-subset of $\F_{q^m}$ and let denote by $\tilde{\mathcal{P}}_{r,\sigma}$ the subset of $\mathcal{P}_{r,\sigma}$ whose polynomials have kernel contained in $U_{\cS}$.
There exists a subset $\mathcal{F} \subset \tilde{\mathcal{P}}_{r,\sigma}$ of $\sigma$-subspace polynomials whose elements have $\sigma$-degree $r$, and agree on the last $g$ $\sigma$-coefficients, such that
\[ |\mathcal{F} | \geq \frac{{n \brack r}_q}{q^{m(g-1)}}. \]
\end{lemma}
\begin{proof}
Clearly, $|\tilde{\mathcal{P}}_{r,\sigma}|={n \brack r}_q$.
We can partition $\tilde{\mathcal{P}}_{r,\sigma}$ into $q^{m(g-1)}$ subsets according to their last $g$ $\sigma$-coefficients. Then, by applying the pigeonhole principle, there exists $\mathcal{F} \subseteq \tilde{\mathcal{P}}_{r,\sigma}$ as in the assertion.
\end{proof}

In particular, when $n \mid m$, we can take as $U_{\cS}$ the subfield $\F_{q^n}$ of $\F_{q^m}$. In this case we can explicitly exhibit a set of $\sigma$-subspace polynomials agreeing on their top-most $\sigma$-coefficients with exponential size in the value $n$. Toward this aim we briefest the following

\begin{lemma}\label{lemma:numrootsTrace}
Let $t,n$ and $m \in \mathbb{Z}^+$ be positive integers such that $t \mid n$ and $n \mid m$.
Consider the $\sigma$-polynomial
\[f(x)=\sum_{i=0}^{\frac{n}t-1} \beta^{\sigma^{it}} x^{\sigma^{it}} \in \mathcal{L}_{m,\sigma}[X]\]
with $\beta \in \F_{q^n}^*$.
The number of roots of $f$ over $\F_{q^m}$ is $q^{n-t}$.
\end{lemma}
\begin{proof}
Since $f(x)$ has coefficients in $\F_{q^n}$, we may look to the $\F_q$-linear transformation $F\colon \F_{q^{n}}\rightarrow\F_{q^n}$ defined by $F(x)=f(x)$.
Clearly, $F$ is also $\F_{q^t}$-linear (because of the $\sigma$-powers that appear in the expression of $f$), $\dim_{\F_{q^t}} \mathrm{Im}\,F\geq 1$ and $\ker\,F$ corresponds to the roots of $f$ over $\F_{q^n}$.
Note that, if $x_0 \in \F_{q^n}$ then
\[ F(x_0)^{\sigma^t}=\left(\sum_{i=0}^{\frac{n}t-1} \beta^{\sigma^{it}} x_0^{\sigma^{it}}\right)^{\sigma^t}=\sum_{i=0}^{\frac{n}t-1} \beta^{\sigma^{(i+1)t}} x_0^{\sigma^{(i+1)t}}= \]
\[ =\sum_{j=1}^{\frac{n}t} \beta^{\sigma^{jt}} x_0^{\sigma^{jt}}=\beta x_0+\sum_{j=1}^{\frac{n}t-1} \beta^{\sigma^{jt}} x_0^{\sigma^{jt}}=F(x_0), \]
and hence $\mathrm{Im}\,F = \F_{q^t}$.
It follows that the number of roots of $f$ over $\F_{q^n}$ is equal to $q^{n-t}$, since $\dim_{\F_{q^t}} \mathrm{Im} \,F=1$.
Moreover, since the $\sigma$-degree of $f$ is $n-t$, by \cite[Theorem 5]{GQ2009} the number of roots of $f$ over $\F_{q^m}$ is at most $q^{n-t}$.
Then, since $\F_{q^n}\subseteq \F_{q^m}$, the assertion follows.
\end{proof}

The following construction extends \cite[Construction 2]{raviv_2016}.

\begin{proposition}\label{lemma:trace}
Let $t,n$ and $m \in \mathbb{Z}^+$ be positive integers such that $t \mid n$ and $n \mid m$.
The set
\[ \mathcal{T}=\left\{ f_{\beta}:=\sum_{i=0}^{\frac{n}t-1} \beta^{\sigma^{it}-\sigma^{n-t}} x^{\sigma^{it}} \colon \beta \in \F_{q^n}^* \right\}\subset \mathcal{L}_{m,\sigma}[x] \]
is a set of $\sigma$-subspace polynomials whose elements have $\sigma$-degree $n-t$, agree on their last $t$ $\sigma$-coefficients and
\[ |\mathcal{T}| = \frac{q^{n}-1}{q^t-1}. \]
\end{proposition}
\begin{proof}
Since for each non zero $\beta \in \F_{q^n}$, the polynomial $f_\beta$ may be viewed as a linearized polynomial of the form discussed in previous Lemma \ref{lemma:numrootsTrace} (up to multiplying by a suitable power of $\beta$), we have that $\dim_{\F_q} \ker\,f_{\beta}=n-t$ and because of the definition of $f_{\beta}$ it is also monic.
The second part follows by applying Corollary \ref{coro:representation_all}, and from the fact that $f_{\alpha}=f_{\beta}$ for some $\alpha,\beta \in \F_{q^n}^*$ if and only if $(\alpha/\beta)^{\sigma^t-1}=1$.
\end{proof}

\begin{remark}\label{rem:n-ker}
If $P \in \mathcal{L}_{m,q}[x]$ and $\mathcal{S}=\{\alpha_1,\ldots,\alpha_n\}$ is an $n$-subset of $\F_{q^m}$ then, denoting by $c_P=(P(\alpha_1),\ldots,P(\alpha_n))$ we have that
\[ \mathrm{rk} (c_P)=n-\dim_{\F_q} ((\ker\,P) \cap U_{\mathcal{S}}), \]
viewing $P$ as an $\F_q$-linear application from $U_{\mathcal{S}}=\langle \alpha_1,\ldots,\alpha_n\rangle_{\F_q}$ to $\F_{q^m}$.
If $\ker\,P \subseteq U_{\mathcal{S}}$, it follows that
\[ \mathrm{rk} (c_P)=n-\dim_{\F_q} \ker\,P. \]
\end{remark}

We can now state a slightly more general version of \cite[Theorem 1]{wachter-zhe_2013}, which in fact may be applied to generalized Gabidulin codes.
The relevan result may be derived from Lemma \ref{lemma:subspaceset} and arguing as in the proof of Theorem 1 of \cite{wachter-zhe_2013}.

\begin{theorem}\label{thm:listdecGab}
Let $k, n$ and $m \in \mathbb{Z}^+$ such that $k \leq n \leq m$. Let $\mathrm{G}_{n,k,\sigma}$ be a generalized Gabidulin code with minimum distance $d=n-k+1$. Let $\tau$ be an integer such that $\left\lfloor\frac{d-1}{2}\right\rfloor+1 \leq \tau \leq d-1$. Then, there exists a word $c \in \F_{q^m}^n \setminus \mathrm{G}_{n,k,\sigma}$ such that
\[|\mathrm{G}_{n,k,\sigma} \cap B_{\tau}(c)| \geq \frac{{n \brack n-\tau}_q}{q^{m(n-\tau-k)}}.\]
\end{theorem}

Of course, we are interested in the smallest values of $\tau$ for which the expression for these lower bounds {\it grows exponentially} in the integer $n$. For the generalized Gabidulin code $\mathrm{G}_{n,k,\sigma}$ we have that
\[ \frac{{n \brack n-\tau}_q}{q^{m(n-\tau-k)}}\geq \frac{q^{\tau(n-\tau)}}{q^{m(n-\tau-k)}} \geq \]
\[ \geq q^{m(1-\epsilon)}\cdot q^{\tau(m+n)-\tau^2-m(d-\epsilon)}, \]
where $0\leq\epsilon<1$.

Hence, as already computed in \cite[Section III]{wachter-zhe_2013}, if \[ \tau \geq \frac{m+n}{2}-\sqrt{\frac{(m+n)^2}{4} -m(d-\epsilon)},\] where $0\leq\epsilon<1$; then the code $\mathrm{G}_{n,k,\sigma}$ cannot be list decoded efficiently, since we find a word $c$ for which $\mathrm{G}_{n,k,\sigma} \cap B_{\tau}(c)$ has exponential size in $n$.

\medskip
\noindent By using Lemma \ref{lemma:subspaceset}, we can now extend the previous result to other MRD codes in \eqref{eq:general_code_from_poly}. Precisely we have

\begin{theorem}\label{thm:listdecGen}
Let $k, n$ and $m \in \mathbb{Z}^+$ such that $k \leq n \leq m$. Let $\mathrm{C} = \mathrm{H}_{n,k,\sigma}(f_1,f_2)$ be an MRD code as in \eqref{eq:general_code_from_poly} with minimum distance $d=n-k+1$ and $\mathrm{C} \notin \{\mathrm{G}_{n,k}, \mathrm{G}_{n,k,\sigma}\}$. Let $\tau$ be an integer such that $\left\lfloor\frac{d-1}{2}\right\rfloor+1 \leq \tau \leq d-1$. Then, there exists a word $c \in \F_{q^m}^n \setminus \mathrm{C}$ such that
\[|\mathrm{C} \cap B_{\tau}(c)| \geq \frac{{n \brack n-\tau}_q}{q^{m(d-\tau)}}.\]
\end{theorem}
\begin{proof}
Let $\tilde{\mathcal{P}}_{n-\tau,\sigma} \subset \mathcal{L}_{m,q}[x]$ be the set of $\sigma$-subspace polynomials of $\sigma$-degree $n-\tau$ associated with $(n-\tau)$-dimensional subspaces of $\F_{q^m}$ contained in $U_{\mathcal{S}}$.
By Lemma \ref{lemma:subspaceset}, there exists a subset ${\cal F}$ of $\tilde{\mathcal{P}}_{n-\tau,\sigma}$ whose elements agree on the last $n-k-\tau+1=d-\tau$ $\sigma$-coefficients, with cardinality at least
\[ \frac{{n \brack n-\tau}_q}{q^{m(d-\tau)}}. \]
Precisely,
\begin{small}
\[{\cal F}=\left\{ \sum _{i=0}^{l-1} a_i x^{\sigma^i}+b_{l}x^{\sigma^{l}}+\cdots+b_{n-\tau-1}x^{\sigma^{n-\tau-1}}+x^{\sigma^{n-\tau}} \colon \right.\]
\[\left. (a_0,a_1,\ldots,a_{l-1}) \in \mathcal{A}  \right\}, \]
\end{small}
where $\mathcal{A}$ is a subset of $\F_{q^m}^{l-1}$ such that $|\mathcal{A}| \geq \frac{{n \brack n-\tau}_q}{q^{m(d-\tau)}}$, and the $b_j$ are fixed elements of $\F_{q^m}$, where $l=n-\tau-(n-k-\tau+1)=k-1$.

Define ${\cal F}' :=  {\cal F} \circ x^{\sigma}$ and note that the $\sigma$-polynomials of ${\cal F}'$ are not $\sigma$-subspace polynomials, although they still have $q^{n-\tau}$ roots.
Clearly they are of type
\[ a_0x^\sigma+a_1x^{\sigma^2}+\ldots+a_{l-1}x^{\sigma^{l}}+b_{l}x^{\sigma^{l+1}}+\ldots+b_{n-\tau-1}x^{\sigma^{n-\tau}}+x^{\sigma^{n-\tau+1}}. \]

Let $R$ be a polynomial in ${\cal F}'$. Note that $c_R \notin \mathrm{H}_{n,k,\sigma}(f_1,f_2)$, since $$\mathrm{rk}(c_{R})=n-\dim_{\F_q}\ker R = \tau < d.$$
Arguing as in the proof of \cite[Theorem 1]{wachter-zhe_2013}, we have that $c_{R-P} \in \mathrm{H}_{n,k,\sigma}(f_1,f_2)$, for each $P \in \mathcal{F}'$, and so in $B_{\tau}(c_R)$ is contains a subset of codewords of $\mathrm{H}_n(f_1,f_2)$ of size at least
\begin{equation}\label{eq:lowerlimit}
\frac{{n \brack n-\tau}_q}{q^{m(d-\tau)}},
\end{equation}
and the assertion follows.
\end{proof}

It is straightforward to show that in this case a code $\mathrm{C}$ of type $\mathrm{H}_{n,k,\sigma}(f_1,f_2)$ with $\mathrm{C} \notin \{\mathrm{G}_{n,k}, \mathrm{G}_{n,k,\sigma}\}$ cannot be list decoded efficiently at the radius $\tau$ if
\[ \tau \geq \frac{m+n}{2}-\sqrt{\frac{(m+n)^2}{4} -m(d+1-\epsilon)},\]
where $0\leq\epsilon<1$.

\section{More MRD codes not list decodable efficiently at all}

Arguing as in Section IV of \cite{raviv_2016}, and exploiting results of the previous section, we exhibit here infinite families of generalized Gabidulin codes with minimum distance $d$, that cannot be list decoded efficiently at all. More precisely, for any radius greater than or equal to $\left\lfloor \frac{d-1}2 \right\rfloor$ we prove the existence of a word $c$ for which $\mathrm{G}_{n,k,\sigma} \cap B_{\tau}(c)$ has exponential size in $n$. Also, we show infinite families of the other types of MRD codes in \eqref{eq:general_code_from_poly} which are not efficiently list decodable for any radius larger than or equal to $\left\lfloor \frac{d-1}2 \right\rfloor+2$, where $d$ represents their minimum distance.
In what follows we will consider $\underline{\alpha}=(\alpha_1,...,\alpha_n)$, to be an ordered $\F_q$-basis of $\F_{q^n}\subseteq \F_{q^m}$.

Also, we remind here that the symbol $\mathrm{G}_{n,k,\sigma}[\underline{\alpha}]$, denotes the rank-metric code obtained by evaluating each element of $\mathcal{G}_{m,k,\sigma}$ in $\underline{\alpha} \in \F_{q^m}^n$.

\smallskip

\subsection{Infinite families of Generalized Gabidulin codes not list decodable efficiently at all}

In this section we provide some refining of Theorems $4$ of \cite{raviv_2016} and of subsequent results in \cite{raviv_2016} directly ensuing from it.

\begin{theorem}\label{thm:listdecGabatall}
Let $k, \tau, n$ and $m \in \mathbb{Z}^+$ positive integers such that $k \leq n$, $\tau \mid n $ and $n \mid m$.

If $\tau > \left\lfloor \frac{d-1}2 \right\rfloor$, then there exists a word $c \in \F_{q^m}^n \setminus \mathrm{G}_{n,k,\sigma}[\underline{\alpha}]$ such that
\begin{equation}\label{eq:q^n/q^t}
|\mathrm{G}_{n,k,\sigma}[\underline{\alpha}] \cap B_{\tau}(c)| \geq \frac{q^n-1}{q^\tau-1}.
\end{equation}
\end{theorem}
\begin{proof}
By Proposition \ref{lemma:trace}, there exists  a set $\mathcal{T}\subset \mathcal{L}_{m,q}[x] $
of $\sigma$-subspace polynomials whose elements have $\sigma$-degree $n-\tau$, agreeing on the last $\tau$ $\sigma$-coefficients and
\[ |\mathcal{T}| = \frac{q^{n}-1}{q^\tau-1}. \]
As in the proof of \cite[Theorem 1]{wachter-zhe_2013}, we can choose a polynomial $R \in \mathcal{T}$ and prove that for each $P \in \mathcal{T}$ we have that
\[ c_{R-P} \in \mathrm{G}_{k,n,\sigma}[\underline{\alpha}] \cap B_{\tau}(c_R). \]
Hence the assertion.
\end{proof}

Consequently, we have the following result.

\begin{corollary}\label{GGnondecatall}
Let $\tau, n,m \in \mathbb{Z}^+$ such that $\tau \mid n $ and $n\mid m$. Then any generalized Gabidulin code $\mathrm{G}_{n,k,\sigma}[\underline{\alpha}]$ with minimum distance $d=2\tau$, cannot be list decoded efficiently at all.
\end{corollary}
\begin{proof}
The unique decoding radius of any code  $\mathrm{G}_{n,k,\sigma}[\underline{\alpha}]$ indicated in the statement, is
\[ \left\lfloor \frac{d-1}2 \right\rfloor. \]
Also, we have proved in Theorem \ref{thm:listdecGabatall} that we may find at least $(q^n-1)/(q^\tau-1)\sim q^{n/2}$ codewords in a ball of radius $\tau= \left\lfloor \frac{d-1}2 \right\rfloor +1$, centered in a suitable word. Hence the list size of such a word is exponential in $n$. It follows that these codes cannot be list decoded efficiently at all.
\end{proof}

\subsection{Infinite families of codes in $\mathrm{H}_{n}(f_1,f_2)$ almost not list decodable efficiently at all}

We now prove similar results for other examples of MRD codes in \eqref{eq:general_code_from_poly}.

\begin{theorem}\label{thm:listdecGenatall}
Let $k,\tau,n$ and $m \in \mathbb{Z}^+$ such that $k \leq n$, $\tau+1 \mid n$ and $n \mid m$. Let $\mathrm{C} = \mathrm{H}_{n,k,\sigma}(f_1,f_2)[\underline{\alpha}]$ be an MRD code defined as in \eqref{eq:general_code_from_poly} with minimum distance $d=n-k+1$ and $\mathrm{C}$.

If $\tau > \left\lfloor \frac{d-1}2 \right\rfloor$, then there exists a word $c \in \F_{q^m}^n \setminus \mathrm{C}$ such that
\[|\mathrm{C} \cap B_{\tau+1}(c)| \geq \frac{q^n-1}{q^{\tau+1}-1}.\]
\end{theorem}
\begin{proof}
By Proposition \ref{lemma:trace},
\[\mathcal{T}=\left\{\sum_{i=0}^{\frac{n}{\tau+1}-1} \beta^{\sigma^{n-\tau-1}-\sigma^{i(\tau+1)}} x^{\sigma^{i(\tau+1)}} \colon \beta \in \F_{q^n}^* \right\}\subset \mathcal{L}_{m,\sigma}[x]\]
is a set of $\sigma$-subspace polynomials whose elements have $\sigma$-degree $n-\tau-1$, agreeing on the last $\tau+1$ $\sigma$-coefficients and
\[ |\mathcal{T}| = \frac{q^{n}-1}{q^{\tau+1}-1}. \]
Then we consider
\[ \mathcal{T}':=\mathcal{T}\circ x^\sigma, \]
which is a set of $\sigma$-polynomials of degree $n-\tau$, agreeing on the last $\tau+1$ $\sigma$-coefficients and
\[ |\mathcal{T}'|= |\mathcal{T}| = \frac{q^{n}-1}{q^{\tau+1}-1}. \]
Note that the polynomials in $\mathcal{T}'$ have still $q^{n-\tau-1}$ roots.

Let $R$ be a polynomial in ${\cal T}'$. Note that $c_R \notin \mathrm{C}$, since $\mathrm{rk}(c_P)=\tau<d$.
For any $P \in {\cal T}'$ we have that $c_{R-P} \in \mathrm{C}$.
Indeed,
\[\deg_\sigma (R-P)\leq n-\tau-(\tau+1)=n-d-1=k-2<k-1\]
and the coefficient of $\sigma$-degree $0$ is zero, hence $R-P \in \mathcal{H}_{m,k,\sigma}(f_1,f_2)$ and $c_{R-P} \in \mathrm{H}_{n,k,\sigma}(f_1,f_2)[\underline{\alpha}]$.
Now, we show that $c_{R-P} \in B_{\tau+1}(c_R)$ for each $P \in \mathcal{F}'$.
In fact, by Remark \ref{rem:n-ker}, since $\ker P \subseteq \F_{q^n}$ we have that
\begin{small}
\[ \mathrm{rk}(c_R-c_{R-P})=\mathrm{rk}(c_P)\leq n-\dim_{\F_q} \ker P=n-(n-\tau-1)=\tau+1. \]
\end{small}
Therefore, arguing as in Theorem \ref{thm:listdecGab}, we have shown that $B_{\tau+1}(c_R)$ contains a subset of codewords of $\mathrm{C}$ of size at least $\frac{q^n-1}{q^{\tau+1}-1}.$
\end{proof}

\begin{corollary}\label{quasinotdecatall}
Let $\tau, n,m \in \mathbb{Z}^+$ such that $\tau +1 \mid n $ and $n\mid m$. Let $\mathrm{C} = \mathrm{H}_{n,k,\sigma}(f_1,f_2)[\underline{\alpha}]$ be an MRD code defined as in \eqref{eq:general_code_from_poly} with minimum distance $d=2\tau$.  Then $\mathrm{C}$ cannot be list decoded efficiently at any radius greater than or equal to $\left\lfloor \frac{d-1}2 \right\rfloor+2$.
\end{corollary}
\begin{proof}
The unique decoding radius of any code described in the statement clearly is
\[ \left\lfloor \frac{d-1}2 \right\rfloor =\tau-1. \]
We have just proved in Theorem \ref{thm:listdecGenatall} that we may find at least $(q^n-1)/(q^{\tau+1}-1)\sim q^{n/2}$ codewords in a ball of radius $\tau+1$; hence, this code cannot be list decoded efficiently for any radius greater than or equal to $ \left\lfloor \frac{d-1}2 \right\rfloor+2$.
\end{proof}

We end the section by underlining that, as a direct consequence of Corollary \ref{quasinotdecatall}, we get families of Gabidulin codes that cannot be list decoded efficiently at any radius greater than or equal to $\left\lfloor \frac{d-1}2 \right\rfloor+2$, with parameters not covered by \cite[Theorem 4]{raviv_2016}.

\section{List decodability of rank-metric codes of $\F_{q^n}^m$ containing $G_{2,n}$}

As a consequence of results proven in  \cite{H-TNRR} and \cite{BR}, MRD codes are {\it dense} in the set of all rank-metric codes; although, very few families of such codes are currently known up to equivalence (for details on the equivalence relation defined for rank-metric codes, we refer the reader to \cite{Morrison_2014} and \cite{sheekey_new_2016}).

In order to distinguish rank-metric codes, in \cite{GZ}, the authors introduced an invariant called the \emph{Gabidulin index} of a rank-metric code $\mathcal{C}$; precisely, it is the maximum dimension of a subcode  of $\mathcal{C}$ equivalent to a generalized Gabidulin code.

The fact that almost all MRD codes known so far contain a generalized Gabidulin code of large dimension, together with aforementioned density results motivate the study of the list decodability of rank-metric codes having Gabidulin index greater or equal than two. Indeed, it is clear that any result in this direction would have a deep impact on a really wide class of codes.
In order to do this we first recall the following result very recently obtained by McGuire and Mueller in \cite{McGuireMueller} relying on the results in \cite{teoremone} and \cite{McGuireSheekey}.

\begin{theorem}\cite[Theorem 1.1]{McGuireMueller}
Let $n=(t-1)t+1$ and $f(x)=x^{q^t}-bx^q-ax \in \mathcal{L}_{n,q}[x]$. If
\begin{itemize}
  \item $\N_{q^n/q}(a)=(-1)^{t-1}$;
  \item $b=-a^{\frac{q^n-q}{q^t-1}}$;
  \item $t-1$ is a power of the characteristic of $\F_{q^n}$,
\end{itemize}
then $f$ has $q^t$ roots in $\F_{q^n}$.
\end{theorem}

Therefore, we can derive the following construction.

\begin{corollary}\label{coro:trin}
Let $t$ and $n \in \mathbb{Z}^+$ be positive integers such that $n=(t-1)t+1$ and $t-1$ is a power of the characteristic of $\F_{q^n}$.
The set
\[ \mathrm{Tri}=\left\{ x^{q^t}-bx^q-ax \colon a,b \in \F_{q^n}, \,\, \N_{q^n/q}(a)=(-1)^{t-1} \right.\]
\[\left. \text{and}\,\, b=-a^{\frac{q^n-q}{q^t-1}} \right\} \subset {\mathcal L}_{n,q}[x]  \]
is a set of $\frac{q^n-1}{q-1}$ subspace polynomials whose elements have $q$-degree $t$.
In particular, if $n \mid m$, then $\mathrm{Tri}$ can be also seen as a set of $\frac{q^n-1}{q-1}$ subspace polynomials of ${\mathcal L}_{m,q}[x]$ whose elements have $q$-degree $t$ over $\F_{q^m}$.
\end{corollary}
\begin{proof}
By the previous result it follows that the polynomials in $\mathrm{Tri}$ are subspace polynomials and the cardinality of $\mathrm{Tri}$ exactly coincides with the number of elements of $\F_{q^m}$ with norm $(-1)^{t-1}$.
The second part follows from the fact that $\F_{q^n}\subseteq\F_{q^m}$ and by \cite[Theorem 5]{GQ2009}.
\end{proof}

By using the previous techniques we are able to prove the following main result.

\begin{theorem}\label{thm:RMcontGn,2}
Let $n, m \in \mathbb{Z}^+$ be positive integers such that $n \mid m$. Let $\mathcal{C}$ be a rank-metric code of $\mathcal{L}_{m,q}[x]$ and let $\mathrm{C}$ be the associated evaluation code over an $\F_q$-basis of $\F_{q^n}$, with minimum distance $d$. Let $\tau \in \mathbb{Z}^+$ such that: \begin{itemize}  \item[1.] $\left\lfloor\frac{d-1}{2}\right\rfloor+1 \leq \tau \leq d-1$ \\ \item[2.] $n-\tau-1$ is a power of the characteristic of $\F_{q^n}$; \\ \item[3.] $n=(n-\tau)(n-\tau-1)+1$. \end{itemize}
Assume that $\mathcal{C}$ contains $\mathcal{G}_{m,2}$.
Then, there exists a word $c \in \F_{q^m}^n \setminus \mathrm{C}$ such that
\[|\mathrm{C} \cap B_{\tau}(c)| \geq \frac{q^n-1}{q-1}.\]
\end{theorem}
\begin{proof}
First note that since $\mathcal{C}\supseteq \mathcal{G}_{n,2}$ then $\mathrm{C}\supseteq \mathrm{G}_{m,2}$.
By Corollary \ref{coro:trin}, the set
\[ \mathrm{Tri}=\left\{ x^{q^{n-\tau}}-bx^q-ax \colon a,b \in \F_{q^n}, \,\, \N_{q^n/q}(a)=(-1)^{n-\tau-1} \right.\]\[\left.\,\,\text{and}\,\, b=-a^{\frac{q^n-q}{q^{n-\tau}-1}} \right\} \]
is a set of $\frac{q^n-1}{q-1}$ subspace polynomials over $\F_{q^m}$ with $q$-degree $n-\tau$.
Let $R \in \mathrm{Tri}$. Since $\deg_{q} R=n-\tau$ and $\ker R \subseteq \F_{q^n}$, by Remark \ref{rem:n-ker} we have that $\mathrm{rk}(c_R)=\tau<d$ and hence $c_R \notin \mathrm{C}$.

Let $P \in \mathrm{Tri}$, then
\[ \deg_{q}(R-P)\leq 1 \]
and so $R-P \in \mathcal{G}_{m,2}\subseteq\mathcal{C}$ and $c_{R-P}\in \mathrm{C}$.
Also, for each $P \in \mathrm{Tri}$, by Remark \ref{rem:n-ker}, since $\ker P \subseteq \F_{q^n}$ it follows
\[ \mathrm{rk}(c_R-c_{R-P})=\mathrm{rk}(c_P)= \]
\[ = n-\dim_{\F_q} \ker\,P=n-(n-\tau)=\tau. \]
Therefore, for each $P \in \mathrm{Tri}$ we have that
\[ c_{R-P} \in \mathrm{G}_{n,2} \cap B_{\tau}(c_R)\subseteq \mathrm{C}\cap  B_{\tau}(c_R). \]
Finally, we have to prove that different choices of $P \in \mathrm{Tri}$ produces different codewords $c_{R-P}$ of $\mathrm{C}$.
Indeed, suppose that $P,P' \in \mathrm{Tri}$ with $P\neq P'$ and $c_{R-P}=c_{R-P'}$, then it follows that
\[ c_{P'-P}=\mathbf{0}, \]
but since $P-P' \in \mathrm{Tran}$, this can not be the case because of Corollary \ref{coro:representation_all}.
This completely proves the assertion.
\end{proof}

\begin{remark}\label{rem:ntau}
Once we fix the integer $n$, providing $4n-3$ is a square in $\mathbb{Z}^+$, we may always find an integer $\tau$ such that
\[ n=(n-\tau)(n-\tau-1) +1,\]
in fact, $\displaystyle \tau=\frac{2n-1-\sqrt{4n-3}}{2}$.
\end{remark}

\begin{remark}\label{rm:size}
We observe that if $\mathrm{C}=\mathrm{G}_{n,k}$ with $k\geq 2$ and with constraints on the involved parameters as prescribed in Theorem \ref{thm:RMcontGn,2}, then there exists a word $c \in \F_{q^m}^n \setminus \mathrm{C}$ such that
\[|\mathrm{C} \cap B_{\tau}(c)| \geq \frac{q^n-1}{q-1} \sim q^{n-1},\] which improves the list size provided in \cite[Theorem 3]{raviv_2016}, for any value of $\tau$ greater than or equal to $\frac{2n-1-\sqrt{4n-3}}{2}$.
\end{remark}

Hence, as a corollary of Theorem \ref{thm:RMcontGn,2} and by applying Remark \ref{rem:ntau}, we have the following result.

\begin{corollary}\label{coro:genGm2}
Let $n, m \in \mathbb{Z}^+$ such that \begin{itemize} \item $n\mid m$; \item $4n-3$ is a square in $\mathbb{Z}$ and $\displaystyle \tau=\frac{2n-1-\sqrt{4n-3}}{2}$;  \item $n-\tau-1$ is a power of the characteristic of $\F_{q^n}$. \end{itemize}

Let $\mathcal{C}$ be an  rank-metric code of $\mathcal{L}_{m,q}[x]$ containing $\mathcal{G}_{m,2}$ and let $\mathrm{C}$ be the associated evaluation code over a basis of $\F_{q^n}$ with minimum distance $d$. Then the code $\mathrm{C}$ is not $\tau$-list decodable efficiently.
In particular, it is not $t$-list decodable efficiently for each $t\geq \tau$.
\end{corollary}

We conclude the section by showing an example.

\begin{example}
Suppose that $n=7$ and $q$ is even. By Remark \ref{rem:ntau} we have that $\tau=4$.
Let $\mathcal{C}$ be an MRD code of $\mathcal{L}_{m,q}$, with $m=7\cdot \ell$ and $|\mathcal{C}|=q^{3m}$ and so $\mathrm{C}$ is an MRD code of $\F_{q^m}^7$ with minimum distance $d=5$.
By the previous corollary, such an MRD code is not $4$-list decodable efficiently.
\end{example}

\section{Conclusions and final remarks}

Applying the puncturing operation to a Delsarte-Gabidulin code $\mathrm{G}_{n,k}$, under some strict constraints on involved parameters, in \cite{raviv_2016},  the authors succeeded in showing the existence of a Delsarte-Gabidulin code $\mathrm{G}_{n-1,k}$ not list decodable efficiently at any value beyond the unique decoding radius. As a consequence of Theorems \ref{thm:listdecGab} and \ref{thm:listdecGen} of Section \ref{sec:bounds}, same approach as that used in Lemma 7 and  subsequent Corollary 3 of \cite{raviv_2016}, leads to similar achievements for generalized Gabidulin codes and for the other examples in Class \eqref{eq:general_code_from_poly}. This strengthens the belief that divisibility condition between $n$ and $m$ may be actually ruled out.

Nonetheless, Theorems \ref{thm:listdecGab} and \ref{thm:listdecGen} have a counterpart for {\it subspace codes} naturally associated with the MRD codes in the Class \eqref{eq:codes} by means of the so called {\it lifting} procedure. Precisely, providing the radius is large enough, the list associated to certain subspaces of the vector space $V(m+n,q)$ turns out to be exponential, which makes these subspace codes not efficiently list decodable, as well. Also, values of the parameters can be introduced in order to get examples that can not be list decoded efficiently at all.

The  behavior of the codes in \eqref{eq:general_code_from_poly} from the list decodability point of view does not rule out the possibility to find out subcodes of relevant codes, for which efficient algorithms for list decoding exists, whenever a reasonable amount of rank error beyond the unique decoding radius occur. For instance in \cite{guruswami_2016}, under the hypothesis that $n\mid m$, the authors provide  a subcode of a Delsarte-Gabidulin code $\mathrm{G}_{n,k}$ that in fact can be list decoded efficiently up to $\frac{s(n-k)}{s+1}$ errors, where $s$ is any integer such that $1\leq s \leq m$.

We point out that techniques developed in Section $IV$ of \cite{guruswami_2016}, specifically Lemmas $14$ and $15$ and $16$, which are key tools towards the determination of relevant subcodes and related list decoding algorithm, may be adapted to codes in \eqref{eq:general_code_from_poly}, still providing divisibility condition.

Finally, one interesting problem to be addressed for future research is studying in which circumstances, if there are, it is possible to generalize results of \cite{McGuireMueller}  to $\sigma$-polynomials. In fact, this will yield to a generalization of Theorem \ref{thm:RMcontGn,2} and Corollary \ref{coro:genGm2} to any rank-metric code of Gabidulin index two.


\begin{IEEEbiographynophoto}{Rocco Trombetti}
was born in Caserta (Italy) in 1975. He received the Degree in Mathematics in 1997 from the University of Campania “Luigi Vanvitelli”, and the Ph.D in Mathematics in 2004 from the University of Naples ``Federico II", where he is currently a Professor. His research interests are in combinatorics, with particular regard to finite geometry. He obtained results, in collaboration also with Italian and foreign researchers, on the following topics: spreads and ovoids of polar spaces, semifields, non-associative algebras and associated geometric structures, MRD-codes.
\end{IEEEbiographynophoto}

\begin{IEEEbiographynophoto}{Ferdinando Zullo}
was born in Piedimonte Matese (Caserta), Italy in 1991. He received the B.S.\ degree, M.S.\ degree and PhD degree in mathematics from the University of Campania “Luigi Vanvitelli” in Caserta, Italy respectively, in 2013, 2015 and 2018.
Currently he has a postdoctoral research fellow at the University of Campania “Luigi Vanvitelli” in Caserta, Italy. His research interests include finite fields, rank metric codes, incidence structures, linear codes, blocking sets and linear sets.
\end{IEEEbiographynophoto}





\end{document}